	\documentclass[
		paper=a4,11pt,american,pagesize,%
		abstract=true,headinclude=true,headlines=0,footinclude=true,footlines=-5,twoside=semi,%
		DIV=12,%
	]{scrartcl}
	\def\docclass{koma}
	\def\version{arxiv}

	\def\draftmode{false} 

\usepackage[utf8]{inputenc}
\makeatletter
\usepackage{xifthen}

\newcommand\iflipics[2]{\ifthenelse{\equal{\docclass}{lipics}}{#1}{#2}}
\newcommand\ifkoma[2]{\ifthenelse{\equal{\docclass}{koma}}{#1}{#2}}
\newcommand\ifieee[2]{\ifthenelse{\equal{\docclass}{ieee}}{#1}{#2}}
\newcommand\ifsiam[2]{\ifthenelse{\equal{\docclass}{siam}}{#1}{#2}}
\newcommand\ifsiamsingle[2]{\ifthenelse{\equal{\docclass}{siam-single}}{#1}{#2}}
\newcommand\ifmysiam[2]{\ifthenelse{\equal{\docclass}{my-siam}}{#1}{#2}}
\newcommand\ifacm[2]{\ifthenelse{\equal{\docclass}{acm}}{#1}{#2}}
\newcommand\ifdcc[2]{\ifthenelse{\equal{\docclass}{dcc}}{#1}{#2}}
\newcommand\ifspringerjournal[2]{\ifthenelse{\equal{\docclass}{springer-journal}}{#1}{#2}}
\newcommand\iflncs[2]{\ifthenelse{\equal{\docclass}{lncs}}{#1}{#2}}
\ifthenelse{ \equal{\docclass}{lipics} \OR \equal{\docclass}{koma} \OR \equal{\docclass}{ieee} \OR \equal{\docclass}{siam} \OR \equal{\docclass}{my-siam} \OR \equal{\docclass}{acm} \OR \equal{\docclass}{dcc} \OR \equal{\docclass}{springer-journal} \OR \equal{\docclass}{lncs} }{
	\PackageInfo{paper}{Building paper with docclass = \docclass} 
}{
	\PackageWarning{paper}{docclass = "\docclass", but must be one of "lipics", "koma", "ieee", "siam", "siam-single", "my-siam", "acm", "dcc", "springer-journal", "lncs"}
}

\newcommand\ifmanuscript[2]{\ifthenelse{\equal{\version}{manuscript}}{#1}{#2}}
\newcommand\ifarxiv[2]{\ifthenelse{\equal{\version}{arxiv}}{#1}{#2}}
\newcommand\ifsubmission[2]{\ifthenelse{\equal{\version}{submission}}{#1}{#2}}
\newcommand\ifproceedings[2]{\ifthenelse{\equal{\version}{proceedings}}{#1}{#2}}
\ifthenelse{ 
	\equal{\version}{manuscript} 
	\OR \equal{\version}{arxiv} 
	\OR \equal{\version}{submission} 
	\OR \equal{\version}{proceedings} 
}{
	\PackageInfo{paper}{Building paper version = \version} 
}{
	\PackageWarning{paper}{version = "\version", but must be one of "manuscript", "arxiv", "submission", "proceedings"}
}

\newcommand\ifdraft[2]{\ifthenelse{\equal{\draftmode}{true}}{#1}{#2}}
\ifthenelse{ \equal{\draftmode}{true} \OR \equal{\draftmode}{false} }{
	\PackageInfo{paper}{Building paper with draftmode = \draftmode} 
}{
	\PackageWarning{paper}{draftmode = "\draftmode", but must be "true" or "false"}
}

\usepackage[T1]{fontenc}
\ifsiam{
	\usepackage{lmodern}
	\usepackage{slantsc}
}{}
\ifsiamsingle{
	\usepackage{lmodern}
	\usepackage{slantsc}
}{}
\ifmysiam{
	\usepackage{lmodern}
	\usepackage{slantsc}
}{}
\ifkoma{
	\usepackage{lmodern}
	\usepackage{slantsc}
}{}
\iflipics{
	\usepackage{lmodern}
	\usepackage{slantsc}
}{}
\ifdcc{
	\usepackage{lmodern}
	\usepackage{slantsc}
}{}
\ifspringerjournal{
	\usepackage{lmodern}
	\usepackage{slantsc}
	\usepackage{xcolor}
}{}
\iflncs{
	\usepackage{lmodern}
	\usepackage{slantsc}
	\usepackage{xcolor}
}{}

\usepackage{babel}
\input{ushyphex.tex} 

\usepackage{array,multicol,multirow}
\ifieee{
	\usepackage[cmex10]{amsmath,mathtools}
	\usepackage{amsfonts,amssymb}
}{}
\ifkoma{
	\usepackage{amsmath,amsfonts,amssymb,mathtools}
}{}
\iflipics{
	\usepackage{amsmath,amsfonts,amssymb,mathtools}
}{}
\ifsiam{
	\usepackage{amsmath,amsfonts,amssymb,mathtools}
}{}
\ifsiamsingle{
	\usepackage{amsmath,amsfonts,amssymb,mathtools}
}{}
\ifmysiam{
	\usepackage{amsmath,amsfonts,amssymb,mathtools}
}{}
\ifacm{
	\usepackage{mathtools}
}{}
\ifdcc{
	\usepackage{amsmath,amsfonts,amssymb,mathtools}
}{}
\ifspringerjournal{
	\usepackage{amsmath,amsfonts,amssymb,mathtools}
}{}
\iflncs{
	\usepackage{amsmath,amsfonts,amssymb,mathtools}
}{}

\usepackage{mleftright}\mleftright 
\usepackage{relsize,xspace,booktabs,adjustbox,needspace,pbox,relsize,makecell}
\ifieee{
		
		\usepackage{enumitem}
}{
	\ifspringerjournal{
		\usepackage{enumitem}
		\setlist{topsep=\medskipamount}
	}{
		\usepackage{enumitem}
	}
}
\usepackage{graphicx}

\ifacm{}{
	\usepackage{colonequals}
}

\usepackage{wref}

\ifsiam{
	\usepackage{ltexpprt}
}{}
\ifsiamsingle{
	\usepackage{ltexpprt}
}{}

\newdimen\makeboxdimen

\ifthenelse{\equal{\docclass}{koma} \OR \equal{\docclass}{my-siam}}{
	\setlength\parindent{1.5em}
	\usepackage[headsepline]{scrlayer-scrpage}
	\pagestyle{scrheadings}
	\clearscrheadfoot
	\AtBeginDocument{%
		\automark[section]{}%
	}
	\ohead{\pagemark}
	\rehead{\mytitle}
	\lohead{\mytitle}
	\addtokomafont{caption}{\sffamily\small}
	\addtokomafont{captionlabel}{\sffamily\textbf}
	\setcapmargin{2em}
}{}
\ifmysiam{
	\setcapmargin{1em}
	\setcapindent{0em}
}{}
\ifmysiam{
	\setlength\parskip{0pt}
	\RedeclareSectionCommand[
		beforeskip=-1.25\baselineskip,
		afterskip=0.75\baselineskip,
	]{section}
	\RedeclareSectionCommand[
		beforeskip=-1\baselineskip,
		afterskip=-1.5em,
	]{subsection}
	\RedeclareSectionCommand[
		beforeskip=-1\baselineskip,
		afterskip=-1.5em,
	]{subsubsection}
	\RedeclareSectionCommand[
		beforeskip=-.25\baselineskip,
		indent=1.5em,
		afterskip=-1em,
	]{paragraph}
}{}
\ifdcc{
	\ifproceedings{}{
		\pagestyle{plain}
		\setlength{\footskip}{5ex}
	}
}{}

\AtBeginDocument{%
	\let\mytitle\@title%
}

\newcommand\shorttitle[1]{%
	\ifacm{%
		\FAIL
	}{}
	\ifspringerjournal{%
		\FAIL
	}{}
	\iflipics{%
		\titlerunning{#1}%
	}{}%
	\iflncs{%
		\titlerunning{#1}%
	}{}%
	\AtBeginDocument{%
		\def\mytitle{#1}%
	}%
}

\ifsiam{
	\usepackage[bibtex-alpha]{url-doi-arxiv}
}{}
\ifsiamsingle{
	\usepackage[bibtex-alpha]{url-doi-arxiv}
}{}
\ifmysiam{
	\usepackage[bibtex-alpha]{url-doi-arxiv}
}{}
\ifkoma{
	\usepackage[bibtex-alpha]{url-doi-arxiv}
}{}
\iflipics{
	\usepackage[bibtex]{url-doi-arxiv}
}{}
\ifdcc{
	\usepackage[bibtex-alpha]{url-doi-arxiv}
}{}
\ifspringerjournal{
	\usepackage[bibtex-alpha]{url-doi-arxiv}
}{}
\iflncs{
	\usepackage[bibtex-alpha]{url-doi-arxiv}
}{}

\let\oldthebibliography\thebibliography
\renewcommand\thebibliography[1]{%
	\oldthebibliography{#1}%
	\pdfbookmark[1]{References}{}%
}

\usepackage{lscape} 

\ifkoma{
	\usepackage{float}
	\floatstyle{plain}
	\usepackage{newfloat}
	\DeclareFloatingEnvironment[%
			name=Algorithm,%
			placement=thb,%
		]{algorithm}
}{}
\iflipics{
	\usepackage{newfloat}
	\DeclareFloatingEnvironment[%
			name=Algorithm,%
			placement=thb,%
		]{algorithm}
}{}
\ifmysiam{
	\usepackage{newfloat}
	\DeclareFloatingEnvironment[%
			name=Algorithm,%
			placement=thb,%
		]{algorithm}
}

\ifdcc{
	\usepackage{float}
	\usepackage[font={small},labelfont=bf]{caption}
}{}
\ifmysiam{
	\setcounter{topnumber}{3}
	\setcounter{bottomnumber}{3}
	\setcounter{totalnumber}{3}     
	\setcounter{dbltopnumber}{3}    

}{}
\ifsiam{

	\setcounter{topnumber}{3}
	\setcounter{bottomnumber}{3}
	\setcounter{totalnumber}{3}     
	\setcounter{dbltopnumber}{3}    

}{}
\ifsiamsingle{

	\setcounter{topnumber}{3}
	\setcounter{bottomnumber}{3}
	\setcounter{totalnumber}{3}     

}{}

\usepackage{dcolumn}

\usepackage{wclrscode}

\usepackage{textcomp} 
\usepackage{listings}

\usepackage{tikz}

\usetikzlibrary{positioning,arrows.meta,fit}
\usetikzlibrary{backgrounds,calc,trees,graphs}
\usetikzlibrary{shapes.geometric,shapes.misc}
\usetikzlibrary{decorations,decorations.pathreplacing}
\usetikzlibrary{scopes}

\pgfdeclarelayer{background}
\pgfsetlayers{background,main}

\usetikzlibrary{external}
\tikzsetexternalprefix{pics/externalized/}
\tikzset{
	external/system call={%
		lualatex \tikzexternalcheckshellescape -halt-on-error %
			-interaction=batchmode -jobname "\image" "\texsource"%
	},
}
\tikzset{external/export=false} 

\iflipics{
	\newtheorem{fact}[theorem]{Fact}

	\newenvironment{proofof}[1]{%
		\begin{proof}[{{Proof of #1{}}}]%
	}{%
		\end{proof}%
	}
}{}
\ifacm{
	\AtEndPreamble{
		\theoremstyle{acmdefinition}
		\newtheorem{remark}[theorem]{Remark}
		\newtheorem{fact}[theorem]{Fact}
	}
	
	\newenvironment{proofof}[1]{%
		\begin{proof}[{{Proof of #1{}}}]%
	}{%
		\end{proof}%
	}
}{}
\ifsiam{
	\newtheorem{remark}{Remark}
	\newenvironment{proofof}[1]{%
		\begin{proof}[{{#1{}}}]%
	}{%
		\end{proof}%
	}
}{}
\ifsiamsingle{
	\newtheorem{remark}{Remark}
	\newenvironment{proofof}[1]{%
		\begin{proof}[{{#1{}}}]%
	}{%
		\end{proof}%
	}
}{}
\iflncs{
	\spnewtheorem{fact}[theorem]{Fact}{\itshape}{}
	
	\let\orig@endproof\endproof
	\def\endproof{\qed\orig@endproof}
	\newenvironment{proofof}[1]{%
		\begin{proof}[{{#1{}}}]%
	}{%
		\end{proof}%
	}
}{}
\ifthenelse{%
		\equal{\docclass}{lipics} \OR \equal{\docclass}{siam} \OR 
		\equal{\docclass}{siam-single} \OR \equal{\docclass}{acm} \OR
		\equal{\docclass}{lncs}%
}{}{
	\usepackage[amsmath,hyperref,thmmarks]{ntheorem}
	
	\theoremheaderfont{\sffamily\upshape\bfseries}
	\theorembodyfont{\slshape}
	\theoremseparator{:}
	\newtheoremstyle{proofstyle}%
	  {\item[\theorem@headerfont\hskip\labelsep ##1\theorem@separator]}%
	  {\item[\theorem@headerfont\hskip\labelsep ##3\theorem@separator]}
	
	\theorempreskip{\topsep} 

	\theoremsymbol{\adjustbox{scale=.8}{$\triangleleft\mkern-1mu$}}
	
	\newtheorem{theorem}{Theorem}[section]
	
	\theoremstyle{plain}
	\theorempreskip{\topsep}

	\newtheorem{lemma}[theorem]{Lemma}

	\theoremstyle{plain}
	\theorembodyfont{\upshape}

	\theoremsymbol{\raisebox{-.25ex}{$\Box$}}
	\qedsymbol{\raisebox{-.25ex}{$\Box$}}
	
	\theoremstyle{proofstyle}
	\theoremheaderfont{\sffamily\slshape}
	\newtheorem{proof}{Proof}

}

\iflipics{
		\newenvironment{thmenumerate}[2][]{%
			\begin{enumerate}[
				label={\textsf{\textbf{\color{darkgray}{\makebox[\widthof{(a)}][c]{\textup{(\alph*)}}}}}},
				ref={\ref{#2}\kern.1em--\kern.1em(\alph*)},
				itemsep=0pt,
				topsep=.5ex,
				leftmargin=1.75em,
				#1
			]%
		}{%
			\end{enumerate}%
		}
}{
	\ifspringerjournal{
		\newenvironment{thmenumerate}[2][]{%
			\begin{enumerate}[
				label={\makebox[\widthof{(a)}][c]{\textup{(\alph*)}}},
				ref={\ref{#2}\kern.1em--\kern.1em(\alph*)},
				itemsep=0pt,
				topsep=\smallskipamount,
				leftmargin=1.75em,
				#1
			]%
		}{%
			\end{enumerate}%
		}
	}{
		
	}
}

\newcommand*\ie{\mbox{i.\hspace{.2ex}e.}}
\newcommand*\eg{\mbox{e.\hspace{.2ex}g.}}

\newcommand\N{\mathbb N}

\usepackage{fixmath}

\newcommand{\ESymbol}{\mathbb{E}}

\newcommand{\ProbSymbol}{\ensuremath{\mathbb{P}}}

\DeclarePairedDelimiterXPP\Prob[1]{\ProbSymbol}[]{}{%
	#1%
}
\DeclarePairedDelimiterXPP\E[1]{\ESymbol}[]{}{%
	#1%
}
\DeclarePairedDelimiterXPP\Eover[2]{\ESymbol_{#1}}[]{}{%
	#2%
}
\DeclarePairedDelimiterXPP\ProbIn[2]{\ProbSymbol_{#1}}[]{}{%
	#2%
}
\providecommand{\Prob}{} 
\providecommand{\ProbIn}{} 
\providecommand{\E}{} 
\providecommand{\Eover}{} 

\newcommand{\surroundedmath}[3]{
	\mathchoice{
		#1{#2{#3}#2}%
	}{
		#1{#3}%
	}{
		#1{#3}%
	}{
		#1{#3}%
	}%
}

\newcommand\wrel[1]{\surroundedmath{\mathrel}{\;}{#1}}
\newcommand\wwrel[1]{\surroundedmath{\mathrel}{\;\;}{#1}}
\newcommand\bin[1]{\surroundedmath{\mathbin}{\:}{#1}}

\newcommand{\relwithtext}[3][c]{%
	\mathrel{\underset{\mathclap{\makebox[\widthof{$=$}][#1]{\scriptsize#2}}}{#3}}%
}

\iflipics{}{
	\makeatletter
	\let\oldalign\align
	\let\endoldalign\endalign
	\renewenvironment{align}{%
		\begingroup%
		\let\oldhalign\halign
		\def\halign{%
			\let\oldbreak\\%
			\def\nonnumberbreak{\nonumber\oldbreak*}%
			\def\\{%
				\@ifstar{\nonnumberbreak}{\oldbreak}%
			}%
			\oldhalign%
		}
		\oldalign%
	}{%
		\endoldalign%
		\endgroup%
	}
}
\newcommand*\numberthis[1][]{\stepcounter{equation}\tag{\theequation}}

\allowdisplaybreaks[3]

\newcommand\splitaftercomma[1]{%
  \begingroup
  \begingroup\lccode`~=`, \lowercase{\endgroup
    \edef~{\mathchar\the\mathcode`, \penalty0 \noexpand\hspace{0pt plus .25em}}%
  }\mathcode`,="8000 #1%
  \endgroup
}

\def\mydots{\xleaders\hbox to.5em{\hfill.\hfill}\hfill}
\newlength\tmpLenNotations

\ifdraft{%
	\iflipics{
		\usepackage{lineno} 
	}{
		\usepackage[switch]{lineno} 
	}
	\linenumbers
	\overfullrule=6mm
	
	\usepackage[color,notref,notcite]{showkeys}
	\definecolor{refkey}{gray}{.99}
	\colorlet{labelkey}{green!60!black!60}
	
	\usepackage[inline,nolabel]{showlabels}

	\showlabels{cite}
	\showlabels{citealt}
	\showlabels{citealp}
	\showlabels{citet}
	\showlabels{Citet}
	\showlabels{citep}
	\showlabels{citeauthor}
	\showlabels{Citeauthor}
	\showlabels{citefullauthor}
	\showlabels{citeyear}
	\showlabels{citeyearpar}
	\showlabels{wref}
	\showlabels{wpref}
	\showlabels{wtpref}
	\showlabels{wildref}
	\showlabels{wildpageref}
	\showlabels{wildtpageref}
}{}

\iflipics{
	\ifmanuscript{\hideLIPIcs}{}
	\ifarxiv{\hideLIPIcs}{}
	\ifsubmission{}{\nolinenumbers}
}{}

\ifdraft{}{%
	\usepackage{microtype}
}

\hypersetup{
	final,
	unicode=true, 
	bookmarks=true,
	bookmarksnumbered=true,
	bookmarksdepth=2,
	bookmarksopen=true,
	breaklinks=true,
	hidelinks,
}

\newsavebox\tmpbox

\ifdcc{
	\renewcommand\paragraph{\@startsection{paragraph}{4}{\parindent}
	                                      {\smallskipamount}
	                                      {-1em}%
	                                      {\normalfont\normalsize\bfseries}}
}{}
\iflipics{
	\let\oldparagraph\paragraph
	\renewcommand\paragraph[1]{%
		\oldparagraph*{#1}
	}
}{
	\let\oldparagraph\paragraph
	\renewcommand\paragraph[1]{%
		\oldparagraph{#1.}
	}
}

\ifmysiam{
	\let\oldsubsection\subsection
	\renewcommand\subsection[1]{%
		\oldsubsection{#1.}%
	}
	\let\oldsubsubsection\subsubsection
	\renewcommand\subsubsection[1]{%
		\oldsubsubsection{#1.}%
	}
}{}
\ifsiam{
	\let\oldsubsection\subsection
	\renewcommand\subsection[1]{%
		\oldsubsection{#1.}%
	}
	\let\oldsubsubsection\subsubsection
	\renewcommand\subsubsection[1]{%
		\oldsubsubsection{#1.}%
	}
}{}
\ifsiamsingle{
	\let\oldsubsection\subsection
	\renewcommand\subsection[1]{%
		\oldsubsection{#1.}%
	}
	\let\oldsubsubsection\subsubsection
	\renewcommand\subsubsection[1]{%
		\oldsubsubsection{#1.}%
	}
}{}

\let\epsilon\varepsilon

\def\myacknowledgements{}
\ifkoma{
	
}{}
\ifieee{
	
}{}
\ifsiam{
	
}{}
\ifsiamsingle{
	
}{}
\ifmysiam{
	
}{}
\ifdcc{
	
}{}
\ifacm{
	
}{}
\ifspringerjournal{
	
}{}
\iflncs{
	
}{}

\setlist[description]{font=\boldmath}

\makeatother

\usepackage{hyperref}

\ifacm{
		\title[Short Title]{My Long Paper Title}
}{}
\ifspringerjournal{
	\title[Short Title]{My Long Paper Title}
}{
	\title{~\\[-3\baselineskip]\boldmath Top-Down Mergesort with Sorted Check Has Mergecost $\le (\mathcal H+3)n$}
	\shorttitle{Top-Down Mergesort with Sorted Check Has Mergecost $\le (\mathcal H+3)n$} 
}

\iflipics{
	\author{Sebastian Wild}{University of Marburg, Germany}{wild@informatik.uni-marburg.de}{https://orcid.org/0000-0002-6061-9177}{}
	
	\authorrunning{S. Wild}
	\Copyright{Sebastian Wild}
	

	\keywords{}
}{}

\ifacm{
	\author{Sebastian Wild}
	\orcid{0000-0002-6061-9177} 
	\affiliation{%
		\institution{University of Liverpool}
		\streetaddress{Ashton Building, Ashton Street}
		\city{Liverpool}
		\country{UK}
		\postcode{L69 3BX}
		\position{Lecturer}
		\department{Department of Computer Science}
	}
	\email{wild@liverpool.ac.uk}
	
	\setcopyright{rightsretained} 
	\copyrightyear{2022}
	\acmYear{2022}
	\acmDOI{XXXXXXX.XXXXXXX}
	
	\acmJournal{JACM}

	\keywords{}

	
}{}
\ifspringerjournal{
	
	\author*[1]{\fnm{Sebastian} \sur{Wild}}\email{wild@liverpool.ac.uk}
	
	\affil*[1]{\orgdiv{Department of Computer Science}, \orgname{University of Liverpool}, \orgaddress{\street{Street}, \city{Liverpool}, \postcode{L69 3BX}, \country{UK}}}
}{}
\iflncs{
	\author{%
		Sebasitan Wild%
		\inst{1}%
		\orcidID{0000-1111-2222-3333}%
	}
	\authorrunning{S. Wild et al.}
	\institute{%
		University of Marburg, Germany
		\email{wild@informatik.uni-marburg.de}
	}
	
}{}
\ifthenelse{\NOT \equal{\docclass}{lipics} \AND \NOT \equal{\docclass}{acm} \AND \NOT \equal{\docclass}{springer-journal} \AND \NOT \equal{\docclass}{lncs} }{
	\newcommand\email[1]{\texttt{#1}}
	\author{%
		Sebastian Wild%
			\footnote{University of Marburg, Germany, 
			\email{wild\,@\,informatik.uni-marburg.de}}
	}
	
	\date{\small\today}
}{}

\ifsiam{
	\ifsubmission{\date{}}{}
	
	\fancyfoot[R]{\scriptsize{Copyright \textcopyright\ 20XX\\
	Copyright for this paper is retained by authors}}
	
	\ifsubmission{}{\setcounter{page}{1}}
	\let\oldabstract\abstract\let\oldendabstract\endabstract
	\renewenvironment{abstract}{%
		\noindent\begin{minipage}{\linewidth}%
		\small\setlength{\baselineskip}{9pt}\oldabstract%
	}{%
		\oldendabstract\end{minipage}%
	}
}{}

\begin{document}

\ifacm{}{\maketitle} 

\vspace*{-6ex}

\begin{abstract}
We consider standard top-down recursive Mergesort, 
where we do a single comparison before calling merge to check
if the two recursively sorted subproblems happen to already be correctly ordered.
(If so, we can skip the merging step).
We show for any input $A[0..n)$ of elements consisting of
$r$ \emph{runs} (maximal increasing contiguous subranges in $A$) of respective 
lengths $L_1,\ldots, L_r$,
the mergecost $M$ (the sum of output sizes of all merges) satisfies
$M \le (\mathcal H+3)n$ for 
$\mathcal H = \sum_{i=1}^r (L_i / n) \log_2(n/L_i)$ the runlength entropy.
\end{abstract}

\ifacm{%
	\maketitle%
}{}


\section{Introduction}

\newcommand\overlaps{\mathrel{\raisebox{-.15ex}{\ensuremath{\square}\kern-.6em}\raisebox{.05ex}{\ensuremath{\square}}}}
\newcommand{\Ceil}[1]{\lceil #1 \rceil}
\newcommand{\Floor}[1]{\lfloor #1 \rfloor}

Run-adaptive sorting allows substantial speedups on partially sorted data,
and is widely used in standard library implementations, \eg, in Timsort~\cite{Peters2002} and Powersort~\cite{MunroWild2018}.
A simple counting argument shows that the mergecost of any comparison-based sorting algorithm
is at least $\mathcal Hn$ in the worst case over all inputs with run lengths $L_1,\ldots,L_r$,
and methods such as Peeksort~\cite{MunroWild2018}, Powersort~\cite{MunroWild2018}, and Length-Adaptive Shiversort~\cite{Juge2024} guarantee a mergecost $M \le (\mathcal H+2)n$ 
(with the typical case being closer to the lower bound).

\medskip

\begin{figure}[b]
	\centering\includegraphics[width=\linewidth]{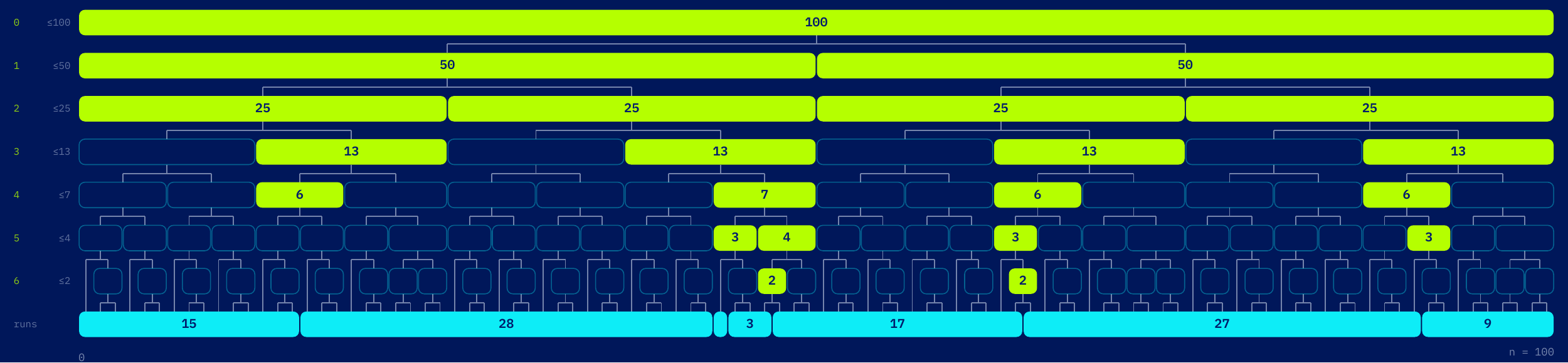}
	\vspace*{-2ex}
	\caption{%
		The recursion tree $T$ of Mergesort on an input of size $n=100$ and run lengths $(L_1,\ldots, L_7) = (15, 28, 1, 3, 17, 27, 9)$.
		Nodes $v\in T$ that do not contain a run boundary are skipped (faded); only filled nodes contribute (their size) to the mergecost $M = 394$.
	}
	\label{fig:recursion}
\end{figure}

\noindent\begin{minipage}[t]{0.52\linewidth}
The above methods explicitly find runs $R_1,\ldots,R_r$ in the input and 
then decide upon a clever order to merge them.
In this note, we show the perhaps-folklore, but unpublished result that the very simple algorithm to the right achieves an only slightly weaker bound.

\medskip

We start by noting that $M$ can be written as
\vspace{-.5ex}
\begin{equation}
\label{eq:M}
		M
	\wrel=
		\sum_{i=1}^r M(R_i),
		\quad M(R_i) \wrel= \sum_{\mathclap{\substack{v\in T:\\ R_i\subsetneq v \,\vee\, v \overlaps R_i}}}  |v\cap R_i|,
\end{equation}
\end{minipage}\hfill
\begin{minipage}[t]{.43\textwidth}
	\begin{lstlisting}[mathescape,gobble=8,basicstyle=\rmfamily]
		procedure Mergesort($A[l..r)$):
			$n \gets r-l$
			if $n \le 1$ return
			$m\gets l + \bigl\lfloor\frac{n}{2}\bigr\rfloor$
			Mergesort($A[l..m)$)
			Mergesort($A[m..r)$)
			if $A[m-1] > A[m]$
				Merge($A[l..m)$, $A[m..r)$, $\mathit{buf}$)
				copy $\mathit{buf}$ to $A[l..r)$
	\end{lstlisting}
\end{minipage}
\medskip

\noindent
where $T$ is the recursion tree of standard Mergesort, $v$ and $R_i$ are identified with their (integer) intervals of indices
contained in their range, and the (nonstandard) notation
$A\overlaps B$ ($A$ ``properly intersects'' $B$) means that $A\cap B \ne \emptyset \wedge A\nsubseteq B \wedge A\nsupseteq B$, neither contained nor disjoint.

\weqref{eq:M} gives a charging scheme for the mergecost: 
each run pays for \emph{its contribution} to the overall mergecost in the merges it participates in.
\wref{fig:charging} gives an example.

\begin{figure}
	\centering\includegraphics[width=\linewidth]{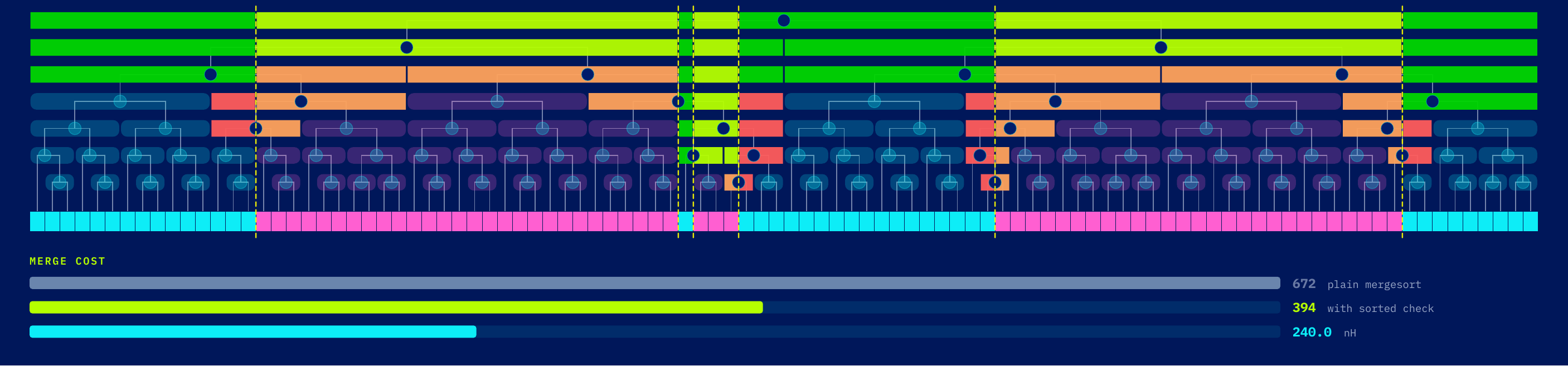}
	\vspace*{-2ex}
	\caption{%
		The recursion tree from \wref{fig:recursion} showing the charging scheme of the analysis.
		Each run pays the colored parts in its vertical range.
	}
	\label{fig:charging}
\end{figure}

A simple calculation shows that it suffices to prove that $M(R_i) \le L_i(\lg(n/L_i) + 3)$, so in the following, we consider $M(R)$ for a single run of length $L$.

\section{Power of Two Case}
Without rounding, \ie, for $n$ a power of 2, the calculation is cleanest.
We account for costs level by level.
At depth $d$, a recursive subproblem has size $n/2^d$. 
We never need to charge $R$ more than $L=|R|$. 
At the same time, there can be at most 2 nodes $v\in T$ (recursive calls) at depth $d$ that contribute to $M(R)$ (cf.\ \wref{fig:charging}), so charging $R$ at most $2\cdot n/2^d$ for depth $d$ also suffices, and we may take the minimum of the two.
Hence%
\footnote{%
	The same calculation appears, \eg, in~\cite[\S3.3]{BrodalWild2023}.
}
\begin{align}
		M(R) 
	&\wwrel\le 
		\sum_{d\ge0} \min\{ L,\, 2\cdot n/2^d \}
	\wwrel{\relwithtext{$(*)$}{\le}}
		L\cdot\lg\frac{n}{L} + 3L
\label{eq:M(R)}
\end{align}
To see why inequality $(*)$ is true, set $G=\lfloor \lg(2n/L)\rfloor$
and $\theta = \{\lg(2n/L)\} = \lg(2n/L) - G \in [0,1)$.
Then $2^\theta = 2n/(L 2^G)$ and $2n/2^G = L  2^\theta$.
Now note that the minimum in $M(R)$ is $L$ iff $d\le G$; otherwise, for $d = G + k$, $k\ge 1$, 
the minimum is 
$2n/2^{d} = (2n/2^{G})\cdot 2^{-k} = L 2^\theta \cdot 2^{-k}$.
Thus
\begin{align*}
		M(R)
	&\wwrel\le
		(G+1)\cdot L \bin+ L\cdot 2^\theta\sum_{k\ge 1} 2^{-k}
	\wwrel=
		L\cdot \bigl(G+1 + 2^\theta\bigr)
\\ &\wwrel=
		L\cdot \left(\lg\frac{n}L  \bin+ g(\theta)\right)
		\qquad\text{with}\quad 
		g(\theta) = 2 - \theta + 2^\theta \in [2.9,3].
\end{align*}

\section{\boldmath General $n$}

For general $n$, some recursive call sizes must be rounded up, so we have to make the analysis
a little tighter. We need two new, but simple ingredients; the first relies on a ceiling fact.

\begin{lemma}[Double Ceiling]
\label{lem:ceil}
For positive integers $n,p,q$, we have:
\(
		\Ceil{\Ceil{n/p}/q} \wwrel= \Ceil{n/(pq)} 
\).
\end{lemma}
\begin{proof}
$\Ceil{n/m}=\min\{k\in\mathbb Z: km\ge n\}$, and for integers $k$ and $n,p,q\in\N$ one has
\\
$kq\ge\Ceil{n/p}\iff kq\ge n/p\iff kpq\ge n$. 
%
\end{proof}
The first new ingredient is a tight bound on the subproblem size for general $n$.
\begin{lemma}[Node Sizes]
\label{lem:sizes}
Every node $v$ at depth $d$ satisfies $|v|\le\Ceil{n/2^d}$.
\end{lemma}
\begin{proof}
Induction on $d$: the root ($d=0$) has size $n=\Ceil{n/2^0}$.
The child of a node of size $\le\Ceil{n/2^{d}}$ has size at most
$\Ceil{\Ceil{n/2^{d}}/2}=\Ceil{n/2^{d+1}}$, using \wref{lem:ceil}. 
\end{proof}
The second ingredient is that when $R\subsetneq v$ or $v\overlaps R$, $v$ must contain at least one element \emph{not} in $R$,
so 
\[
	|v \cap R| 
	\wwrel\le 
	|v|-1 
	\wwrel\le 
	\Ceil{n/2^d} - 1
	\wwrel\le 
	n/2^d.
\]
We can thus still bound $M(R)$ by $\sum_{d\ge 0} \min\{L, 2\cdot n/2^d\}$ and obtain the bound \weqref{eq:M(R)} by the same calculation, which implies $M\le (\mathcal H+3)n$.

\medskip\noindent
Note that the bound $M\le (\mathcal H+3)n$ is asymptotically tight; on input $(L_1,L_2,L_3) = (1,n-2,1)$
Mergesort has $M=3n-O(1)$, whereas $\mathcal H = O(\lg n / n)$.

	\myacknowledgements

\bibliography{references}


\ifdraft{
	\clearpage
	\part*{Notes-to-self}
	\printnotestoself
}{}

\end{document}